\documentclass[sn-mathphys-num]{sn-jnl}

\usepackage{graphicx}%
\usepackage{multirow}%
\usepackage{amsmath,amssymb,amsfonts}%
\usepackage{amsthm}%
\usepackage{mathrsfs}%
\usepackage[title]{appendix}%
\usepackage{xcolor}%
\usepackage{textcomp}%
\usepackage{manyfoot}%
\usepackage{booktabs}%
\usepackage{algorithm}%
\usepackage{algorithmicx}%
\usepackage{algpseudocode}%
\usepackage{listings}%
\usepackage{subfig}%
\usepackage{comment}%
\usepackage{xcolor}%

\theoremstyle{thmstyleone}%
\newcommand\beq{\begin{equation}}%
\newcommand\eeq{\end{equation}}%
\newtheorem{prop}{Proposition}%
\newtheorem{theo}{Theorem}%
\newtheorem{as}{Assumption}%
\DeclareMathOperator{\sech}{sech}%

\begin{document}

\title[Article Title]{Melnikov Method for a Class of Generalized Ziegler Pendulums}

\author[1]{\fnm{Stefano} \sur{Disca}}
\author*[1]{\fnm{Vincenzo} \sur{Coscia}}\email{cos@unife.it}

\affil*[1]{\orgdiv{Department of Mathematics and Computer Science}, \orgname{University of Ferrara}, \orgaddress{\street{Via Machiavelli 30}, \city{Ferrara}, \postcode{44121}, \country{Italy}}}

\abstract{The Melnikov method is applied to a class of generalized Ziegler pendulums. We find an analytical form for the separatrix of the system in terms of Jacobian elliptic integrals, holding for a large class of initial conditions and parameters. By working in Duffing approximation, we apply the Melnikov method to the original Ziegler system, showing that the first non-vanishing Melnikov integral appears in the second order. An explicit expression for the Melnikov integral is derived in the presence of a time-periodic external force and for a suitable choice of the parameters, as well as in the presence of a dissipative term acting on the lower rod of the pendulum. These results allow us to define fundamental relationships between the Melnikov integral and a proper control parameter that distinguishes between regular and chaotic orbits for the original dynamical system. Finally, in the appendix, we present proof of a conjecture concerning the non-validity of Devaney's chaoticity definition for a discrete map associated with the system.}

\keywords{double pendulum; follower force; Melnikov integral; homoclinic intersections; time-periodic perturbations; dissipation; Devaney chaos}


\pacs[MSC Classification]{70K44, 70K55, 34D10, 37C25}

\maketitle

\section{Introduction}
The Poincaré–Melnikov method \cite{Melnikov} represents one of the fundamental analytical tools in the analysis of the chaotic behavior of continuous dynamical systems. The method has found successful applications in several fields, such as solid mechanics \cite{Kuang2001}, celestial mechanics \cite{Xia, Robinson}, dynamics of fluids \cite{Kuang2006}, optics \cite{Kudryashov}, biology \cite{Glendinning}, and oceanography \cite{Maki}. Computation of the Melnikov integral can also be found for simple mechanical systems, such as the double pendulum \cite{Dullin} and non-linear oscillators \cite{Garcia-Margallo}. {Pendulum-like systems particularly show themselves as an example of relatively simple mechanical systems that easily approach chaos. Together with the Melnikov method, other analytical tools may be used to establish the transition from regular to chaotic motion for these systems; for instance, \cite{Szuminski1} presents results about integrability and non-integrability of a double spring pendulum in the framework of the differential Galois theory, while the method of Lagrangian descriptors has been applied in \cite{Lopez} for the classical double pendulum.} Among others, in \cite{Matsuzaki}, the Melnikov method has been applied to the Ziegler pendulum \cite{Ziegler} subject to damping.

{The generalized Ziegler pendulum subject of this work has been defined for the first time in \cite{Polekhin}, where it is analytically proven that the system is integrable for a certain class of parameters, in a Hamiltonian and a non-Hamiltonian case; furthermore, the transition to chaotic motion is numerically shown for a general choice of parameters and initial conditions. Several variants of this system, including the presence of gravity and friction, are studied in \cite{Disca}; in particular, threshold values for the parameters of the system are found both through analytical treatments and numerical simulations in order to distinguish between periodic and chaotic orbits when the system is subject to conservative or dissipative external forces.} The aim of this work is to apply the Melnikov method to the generalized Ziegler pendulum defined in \cite{Polekhin, Disca}, whose dynamics show different features with respect to the classical Ziegler pendulum. In particular, we look for fundamental relationships between the Melnikov integral and a suitable control parameter that can be interpreted as a non-Hamiltonian threshold between regular and chaotic dynamics of the system.

Let us briefly recall the Melnikov method; for the definitions of hyperbolic points, stable and unstable manifolds, we refer the reader to \cite{Guckenheimer}. Let us consider a one-dimensional Hamiltonian system subject to a periodic perturbation:
\beq\label{Hamiltonian_pert}
\dot{\textbf{x}} = \textbf{f}(\textbf{x}) + \varepsilon \textbf{g}(\textbf{x}, t) \,,
\eeq
where
\beq\begin{split}
&\textbf{x} = (q, p) \\
&\textbf{f} = \bigg( \frac{\partial H(q,p)}{\partial p}, -\frac{\partial H(q,p)}{\partial q} \bigg) \\
&\textbf{g} = \big( g_1(q,p,t), g_2(q,p,t) \big), \quad \textbf{g}(q,p, t+T) = \textbf{g}(q,p, t) \,,
\end{split}\eeq
$q$, $p$ are the canonical coordinates and $T>0$. The main idea of the Melnikov method \cite{Guckenheimer, Cencini, Teschl} is to define a proper function (\textit{Melnikov integral}) that quantifies the distance between the stable and unstable manifolds associated with a hyperbolic point. The Melnikov integral for a system of the form \eqref{Hamiltonian_pert} is defined as
\beq
M(t_0; \alpha) = \int_{-\infty}^{+\infty} \textbf{f}[ \textbf{x}_0(t); \alpha ] \cdot \textbf{g}^\perp[ \textbf{x}_0 (t), t+t_0; \alpha ] dt \,,
\eeq
where $\textbf{x}_0(t)$ is the unperturbed trajectory, $\textbf{g}^\perp = (-g_2, g_1)$, $t_0$ is a reference time $t_0 \in [0, T]$, $\alpha$ is a set of scalar parameters and $ \textbf{f} \cdot \textbf{g}$ is the Euclidean dot product. Zeros of $M(t_0; \alpha)$ correspond to homoclinic intersections, i.e., intersections between the two manifolds, implying the onset of chaos around the separatrix by Smale–Birkhoff theorem \cite{Birkhoff, Smale}. In particular, simple zeros of $M(t_0; \alpha)$ ($M = 0$, $\frac{dM}{dt_0} \ne 0$) correspond to transverse intersections, while double zeros of $M(t_0; \alpha)$ ($M = 0, \frac{dM}{dt_0} = 0$) correspond to tangential intersections.

We remark that the Melnikov method is perturbative. In a general sense, denoting by $D(t_0)$ the splitting distance at the time $t_0$ between the stable and unstable manifolds of the system \eqref{Hamiltonian_pert}, one can define the series expansion
\beq\label{pert_series}
D(t_0; \alpha) = M(t_0; \alpha) + \sum_{n=1}^{+\infty} \varepsilon^n M_n(t_0; \alpha) \,.
\eeq
It is not unusual to obtain identically null terms in the previous series; in these cases, one needs to compute the first non-vanishing Melnikov functions in order to state conclusions about homoclinic intersections.

We point out that the Melnikov theorem provides only a sufficient condition for the onset of chaotic behavior; that is, homoclinic intersections imply chaos, while the absence of homoclinic intersections does not imply general regular orbits.

The generalized Ziegler pendulum is a planar mathematical double pendulum consisting of three material points $A$, $B$, $C$ with mass $m_A$, $m_B$, $m_C$, respectively, and not subject to gravity. The point $A$ is held at constant distance from the point $O$ by means of a massless rod of length $l_2$, while the points $B$ and $C$ are positioned at the ends of a second massless rod hinged in $A$ in such a way that $\overline{BA}=l_1$ and $\overline{AC}=l_3$. The system is subject to angular elastic potentials associated with two cylindrical springs placed on the hinges $A$, $O$ and having elastic constants $k_1$, $k_2$, respectively; furthermore, an external follower force of size $F$ acts always direct as the vector $B-C$. The angles $\varphi_1$, $\varphi_2$, associated with the rotation of the lower and upper rod, respectively, are taken as canonical variables with conjugate momenta $v_1$, $v_2$ (see Figure \ref{ziegler_image}).

\begin{figure}[H]
\includegraphics[width=8.5cm]{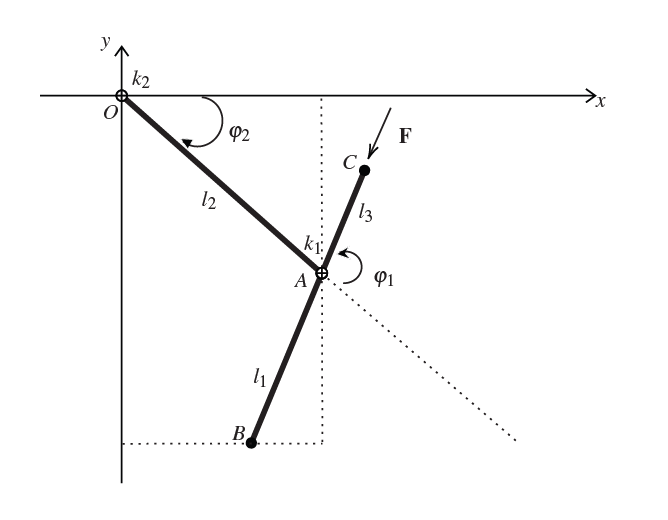}
\caption{The generalized Ziegler pendulum.}
\label{ziegler_image}
\end{figure}

In order to compute the Melnikov integral for the generalized Ziegler pendulum, let us treat it as a Hamiltonian system perturbed by the external follower force $\textbf{F}$. The Hamiltonian of the system is
\beq\label{Hamiltonian}
H(\varphi_1, \varphi_2, v_1, v_2) = \frac{1}{2 A(\varphi_1)} \bigg( A_{22} v_1^2 + A_{11} v_2^2 + 2 A_{12} v_1 v_2 \bigg) + \frac{k_1}{2} \varphi_1^2 + \frac{k_2}{2} \varphi_2^2 \,,
\eeq
where
\beq\label{parameters}
\begin{split}
&A_{11} = m_B l_1^2 + m_C l_3^2 > 0 \\
&A_{12}(\varphi_1) = A_{11} - \Delta l_2 \cos(\varphi_1) \\
&A_{22}(\varphi_1) = A_{11} + M l_2^2 - 2 \Delta l_2 \cos(\varphi_1) \\
&A(\varphi_1) = A_{11} A_{22} - (A_{12})^2 = A_{11} M l_2^2 - \Delta^2 l_2^2 \cos^2(\varphi_1)
\end{split}\eeq
and $M := m_A + m_B + m_C$, $\Delta := m_B l_1 - m_C l_3$. We write the equations of motion in the following perturbative form:
\beq\label{motion}
\begin{cases}
\dot{\varphi_1} = v_1 \\
\dot{v_1} = {-\frac{\partial H}{\partial \varphi_1} (\varphi_1, v_1, \varphi_2, v_2) } + F g_1(\varphi_1) - F \Delta g_2(\varphi_1) \\
\dot{\varphi_2} = v_2 \\
\dot{v_2} = {-\frac{\partial H}{\partial \varphi_2} (\varphi_1, v_1, \varphi_2, v_2) } - F g_1(\varphi_1) \,,
\end{cases}
\eeq
where
\beq\label{perturbations}
g_1(\varphi_1) = \frac{A_{11} l_2}{A(\varphi_1)} \sin(\varphi_1), \quad g_2(\varphi_1) = \frac{g_1(\varphi_1)}{A_{11}} \cos(\varphi_1)
\eeq
and $F$ is a perturbative parameter. If $k_2 = 0$, the variable $\varphi_2$ is cyclic, so that $v_2(t) = \frac{\partial L}{\partial \dot{\varphi_2}} = v_{20}$ is a first integral for the unperturbed Hamiltonian system; moreover, the perturbation depends only on $\varphi_1$. The Equation \eqref{motion} become
\beq\label{motion_k20}
\begin{cases}
\dot{\varphi_1} = v_1 \\
\dot{v_1} = -\frac{\partial H}{\partial \varphi_1} (\varphi_1, v_1; v_{20}) + F g_1(\varphi_1) - F \Delta g_2(\varphi_1) \\
\dot{\varphi_2} = v_{20} \\
\dot{v_2} = - F g_1(\varphi_1) \,,
\end{cases} 
\eeq
so that we can separately apply the Melnikov method to the first subsystem in \eqref{motion_k20} (that is $\varphi_2$-independent) and take into account that $(\varphi_1, v_1) = (0, 0)$ is a hyperbolic point.

The paper is organized as follows. In Section \ref{sec_elliptic}, we derive an analytical expression for the separatrix of the system \eqref{motion} in terms of elliptic integrals under suitable assumptions on initial conditions and parameters. Under further assumptions, in Section \ref{sec_Duffing}, we compute the Melnikov integral of the system for three possible formulations of the problem, i.e., in the presence of a time-independent external force (Section \ref{subsec_Duffing1}), of a time-periodic external force (Section \ref{subsec_Duffing2}), and of a time-periodic external force together with a dissipative term (Section \ref{subsec_Duffing3}). These three formulations are studied separately for the integrable case $\Delta=0$ in Section \ref{sec_Delta0}. The Melnikov function is evaluated at the zero or the first order, depending on the case. In Section \ref{sec_conclusions}, we state our conclusions and propose further developments of this work.

We make a few comments regarding the notation.
\begin{itemize}
\item Where it is not specified, the expression ``Melnikov integral'' refers to the zero-order term in \eqref{pert_series}. 
\item The system \eqref{motion_k20} depends on several scalar parameters; despite this fact, the fundamental parameter we are looking for is $\Delta$, so we write $M(t_0; \Delta)$ and the dependence on the other parameters is understood. For the sake of simplicity in notation, the dependence on all parameters is implied for all other functions defining the dynamical system. 
\item When the Melnikov function does not depend explicitly on time, we simply write $M(\Delta)$.
\end{itemize}

Numerical simulations presented in Sections \ref{sec_Duffing} and \ref{sec_Delta0} show the projection of the motion of the system on the plane $(\varphi_1, v_1)$, as computational support of the analytical results derived in the article. The simulations have been reproduced in MATLAB through the ODE algorithm \textbf{ode45}, by taking a relative and absolute tolerance of $10^{-8}$ and integrating on a time interval of length $10^6$; similar ODE algorithms, such as \textbf{ode113} and \textbf{ode15s}, lead to similar results. Every figure also presents a magnification (The magnifications have been reproduced through a properly modified MATLAB function. Copyright (c) 2016, Kelsey Bower. All rights reserved.) of a specific zone in order to appreciate the details of the motion (up to unavoidable numerical errors).

In the appendix, we present a proof of Proposition 4 conjectured in \cite{Disca}, showing that a discrete map naturally associated with the generalized Ziegler pendulum does not satisfy the definition of chaos in the sense of Devaney \cite{Devaney} for a choice of the parameters associated with chaotic motion for the continuous dynamical system.

\section{Separatrix in Terms of Jacobian Elliptic Integrals}\label{sec_elliptic}
Let us suppose that the first subsystem in \eqref{motion_k20} has a separatrix for $H = E_s$, associated with the hyperbolic point $(\varphi_1, v_1) = (0, 0)$. Starting from the Hamiltonian \eqref{Hamiltonian} and writing $v_1$ as a function of $\varphi_1$, we obtain (here we imply the dependence of $A_{12}$, $A_{22}$, $A$ on $\varphi_1$)
\beq
v_1 = \frac{d \varphi_1}{dt} = \pm \sqrt{-k_1 \frac{A}{A_{22}} \varphi_1^2 - \frac{A}{(A_{22})^2} (v_{20})^2 + \frac{2A}{A_{22}} E_s } + \frac{A_{12}}{A_{22}} v_{20} \,,
\eeq
that leads to
\beq\label{separatrix_gen}
t = \pm \int_{\varphi_{10}}^{\varphi_1(t)} \frac{ A_{22}(\varphi_1) }{ \sqrt{-k_1 A(\varphi_1) A_{22}(\varphi_1) \varphi_1^2 - A(\varphi_1) (v_{20})^2 + 2 A(\varphi_1) E_s } + A_{12}(\varphi_1) v_{20} } d\varphi_1 \,.
\eeq
The previous integral is not reducible to elliptic or hyperelliptic integrals due to the mixing between the quadratic term and trigonometric functions of $\varphi_1$ and the presence of a further term external to the square root at the denominator. One may think of integrating by parts in order to remove rid of the term $\varphi^2$ and deal only with trigonometric functions; however, depending on the integration by parts that one takes, the original mixed term is replaced by logarithmic or polynomials terms of higher order, making the initial problem even more difficult.

In order to simplify the problem, we make the following assumptions.
\begin{as}\label{as1}
$v_{20} = 0$.
\end{as}
\begin{as}\label{as2}
$\varphi_{10}$ is sufficiently small.
\end{as}
\begin{as}\label{as3}
The second-order Taylor polynomial of the denominator in \eqref{separatrix_gen} is factorized into four distinct and real roots.
\end{as}
Under Assumptions \ref{as1}--\ref{as3}, the \eqref{separatrix_gen} can be written in the form
\beq\label{separatrix_gen_v20}
t = \pm \int_{\varphi_{10}}^{\varphi_1(t)} \frac{B_1 + B_2 \cos(\varphi_1') }{ \sqrt{ \varphi_1'^2 ( B_3 \cos^3(\varphi_1') + B_4 \cos^2(\varphi_1') + B_5 \cos(\varphi_1') + B_6 ) + B_7 \cos^2(\varphi_1') + B_8 } } d\varphi_1' \,,
\eeq
where $B_j$, $j = 1, \dots, 8$ are constants that depend on the parameters defined in \eqref{parameters}; now, we may exploit a Taylor expansion on the terms $\cos(\varphi_1)$, in order to reduce the \eqref{separatrix_gen_v20} to Jacobian elliptic functions. This procedure is of course repeatable to an arbitrary order on $\varphi_1$; in particular, by taking the Taylor expansion $\cos(\varphi_1) = 1 - \frac{\varphi_1^2}{2} + \dots + o(\varphi_1^{2n})$, $n \in \mathbb{N}$, the \eqref{separatrix_gen_v20} is written by means of integrals of the form
\begin{subequations}\label{separatrix_gen_v20_Taylor}
\beq\label{separatrix_gen_v20_Taylor1}
J_1 = \int_{\varphi_{10}}^{\varphi_1(t)} \frac{1}{ \sqrt{ Q( \varphi_1'^{2n+2}, \varphi_1'^{2n}, \varphi_1'^{2n-2}, \dots, \varphi_1'^2, 1 ) } } d\varphi_1'
\eeq
\beq\label{separatrix_gen_v20_Taylor2}
J_2 = \int_{\varphi_{10}}^{\varphi_1(t)} \frac{ P( \varphi_1'^{2n}, \varphi_1'^{2n-2}, \dots, \varphi_1'^2 ) }{ \sqrt{ Q( \varphi_1'^{2n+2}, \varphi_1'^{2n}, \varphi_1'^{2n-2}, \dots, \varphi_1'^2, 1 ) } } d\varphi_1' \,,
\eeq
\end{subequations}
where $P, Q$ are polynomial functions of $\varphi_1$. We take $\cos(\varphi) = 1 - \frac{\varphi^2}{2} + o(\varphi^2)$, so that the \eqref{separatrix_gen_v20} becomes
\beq\label{separatrix_gen_v20_explicit}
t = \pm (A_{11} + M l_2^2 - 2 \Delta l_2) \int_{\varphi_{10}}^{\varphi_1(t)} \frac{1}{ \sqrt{ Q(\varphi_1'^4, \varphi_1'^2, 1) } } d\varphi_1' \pm \Delta l_2 \int_{\varphi_{10}}^{\varphi_1(t)} \frac{ \varphi_1'^2 }{ \sqrt{ Q(\varphi_1'^4, \varphi_1'^2, 1) } } d\varphi_1' \,,
\eeq
where
\beq\begin{split}
Q(\varphi_1^4, \varphi_1^2, 1) =& -k_1 \varphi_1^4 \bigg( - 3 l_2^3 \Delta^2 + (A_{11} + M l_2^2) l_2^2 \Delta + A_{11} M l_2^3 \bigg) \Delta - \\
&-k_1 \varphi_1^2 \bigg( 2 l_2^3 \Delta^3 + ( E_s - A_{11} - M l_2^2 ) l_2^2 \Delta^2 - \\
&\qquad\qquad -2 A_{11} M l_2^3 \Delta + A_{11} M l_2^2 (A_{11} + Ml_2^2) \bigg) + \\
&+2 E_s ( A_{11} M - \Delta^2 ) l_2^2 \,.
\end{split}\eeq

Given Assumption \ref{as3}, we denote by $a > b > c > d$ the four roots of $Q$; therefore, $J_1$ is reduced to (see Equation (251.00) in \cite{Byrd})
\beq
J_1 = \frac{2}{\sqrt{ (a-c)(b-d) }} ( F(\theta_0, k) - F(\theta(t), k) ) \,,
\eeq
where
\beq\begin{split}
&k^2 = \frac{ (b-c)(a-d) }{ (a-c)(b-d) } \\
&\theta_0 = \sin^{-1}{ \sqrt{ \frac{ (a-c)(d - \varphi_{10}) }{ (a-d)(c - \varphi_{10}) } } } \\
&\theta(t) = \sin^{-1}{ \sqrt{ \frac{ (a-c)(d - \varphi_1(t) ) }{ (a-d)(c - \varphi_1(t) ) } } } \\
\end{split}\eeq
and $F(\theta, k)$ is the incomplete elliptic integral of the first kind. Analogously, $J_2$ is reduced to (see Equation (251.03) in \cite{Byrd})
\beq
J_2 = -\frac{d^2}{\sqrt{ (a-c)(b-d) }} ( Z_2(u_{\varphi}) - Z_2(u_{\varphi_{10}}) ) \,,
\eeq
where
\beq
Z_2(\tilde u) = \int_0^{\tilde u} \frac{ (1 - \frac{c \alpha^2}{d} \text{sn}^2 u )^2 }{ (1 - \alpha^2 \text{sn}^2 u)^2 } du
\eeq
and
\beq\begin{split}
&\alpha^2 = \frac{a-d}{a-c} \\
&\text{sn}^2 u_{\varphi} = \frac{ (a-c) (d - \varphi) }{ (a - d) (c - \varphi) } \\
&\text{sn}^2 u_{\varphi_{10}} = \frac{ (a-c) (d - \varphi_{10}) }{ (a - d) (c - \varphi_{10}) } \,.
\end{split}\eeq
Finally (see Equations (336.00)--(336.03) and (340.04) in \cite{Byrd}), $Z_2(u_{\varphi})$ and $Z_2(u_{\varphi_{10}})$ can be written as a linear combination of the incomplete elliptic integral of the first kind and Legendre's incomplete elliptic integrals of the second and third kind.

\section{Melnikov Integral in Duffing Approximation}\label{sec_Duffing}
For the following calculations, we keep taking Assumptions \ref{as1} and \ref{as2}, i.e., $v_{20} = 0$ and $\varphi_{10}$ sufficiently small; furthermore, we impose the system to have a separatrix for $H(\varphi_1, v_1; 0) = 0$ and derive proper constraints on the parameters such that this assumption is verified. Starting from the Hamiltonian \eqref{Hamiltonian}, similar calculations as before lead to
\beq\label{separatrix}
v_1 = \frac{d \varphi_1}{dt} = \pm \sqrt{ -k_1 \frac{A(\varphi_1)}{ A_{22}(\varphi_1) } \varphi_1^2 } \,.
\eeq
The Melnikov integral for the system \eqref{motion_k20} reads
\beq\label{Melnikov_integral}
M(t_0; \Delta) = \int_{-\infty}^{+\infty} v_1(t) \bigg( -g_1(\varphi_1(t)) + \Delta g_2(\varphi_1(t)) \bigg) dt \,,
\eeq
where $v_1 = \frac{d \varphi_1}{dt}$ is determined by \eqref{separatrix} for suitable initial conditions.

After a second-order Taylor expansion on $\varphi_1$ for the term $\frac{A(\varphi_1)}{A_{22}(\varphi_1)}$, we obtain
\beq\label{separatrix_simple}
\frac{d \varphi_1}{dt} = \pm \sqrt{ -k_1 \frac{a - \Delta^2 l_2^2}{b - 2 \Delta l_2} \varphi_1^2 \bigg( 1 + \frac{1}{2} \frac{\alpha}{ \frac{a - \Delta^2 l_2^2}{b - 2 \Delta l_2} } \varphi_1^2 \bigg) } \,,
\eeq
where
\beq\begin{split}\label{constraint_parameters}
&a = A_{11} M l_2^2 \\
&b = A_{11} + M l_2^2 \\
&\alpha = \frac{d^2}{d \varphi_1^2} \bigg( \frac{A(\varphi_1)}{A_{22}(\varphi_1)} \bigg) \bigg|_{\varphi_1=0} = \frac{ 2 \Delta l_2 (b - 2 \Delta l_2 - a + \Delta^2 l_2^2) }{(b - 2 \Delta l_2)^2} \,.
\end{split}\eeq
We now impose the following constraints on the parameters:
\beq\label{constraint}
\begin{cases}
\frac{a - \Delta^2 l_2^2}{b - 2 \Delta l_2} < 0 \\
\alpha > 0 \,,
\end{cases} 
\eeq
that are equivalent to
\beq
\begin{cases}
\min \{ \frac{b}{2}, \sqrt{a} \} < \Delta l_2 < \max \{ \frac{b}{2}, \sqrt{a} \} \\
0 < \Delta l_2 < 1 - \sqrt{1 - b + a} \cup \Delta l_2 > 1 + \sqrt{1 - b + a}\,.
\end{cases} 
\eeq
From the definitions of $a$, $b$, we have $M = \frac{b \pm \sqrt{b^2 - 4a} }{2 l_2^2}$; therefore we must have $b \ge 2 \sqrt{a}$, in order to obtain a real positive value for $M$. Furthermore, it must be $b \le 1 + a$. Finally, the constraints imposed on the parameters are
\beq\label{constraint_delta}
\begin{cases}
\sqrt{a} < \Delta l_2 < \frac{b}{2} \\
\Delta l_2 < 1 - \sqrt{1 - b + a} \cup \Delta l_2 > 1 + \sqrt{1 - b + a} \\
b \le 1 +a \,.
\end{cases} 
\eeq

Given the \eqref{constraint_delta}, the motion on the separatrix is reduced to a Duffing oscillator \cite{Duffing}; in particular, after the substitution $\theta = \sqrt{ \frac{\alpha}{2 \big| \frac{a - \Delta^2 l_2^2}{b - 2 \Delta l_2} \big| } } \varphi_1$, \eqref{separatrix_simple} becomes
\beq\label{separatrix_theta}
\frac{d \theta}{dt} = \pm \sqrt{-k_1 \frac{a - \Delta^2 l_2^2}{b - 2 \Delta l_2} } \sqrt{\theta^2 (1 - \theta^2)} \,.
\eeq
The Ansatz $\theta(t) = \pm \sech(x t)$ are solution of the \eqref{separatrix_theta} if
\beq\label{xy}
x = \mp \sqrt{-k_1 \frac{a - \Delta^2 l_2^2}{b - 2 \Delta l_2} } \,.
\eeq
Finally, given the \eqref{xy}, the motion on the separatrix is parametrized by
\beq\label{separatrix_solved}
\begin{cases}
\varphi_1(t) = \pm \sqrt{ \frac{2 \big| \frac{a - \Delta^2 l_2^2}{b - 2 \Delta l_2} \big| }{\alpha} } \sech \bigg( \sqrt{-k_1 \frac{a - \Delta^2 l_2^2}{b - 2 \Delta l_2} } t \bigg) \\
v_1(t) = \mp \sqrt{ \frac{2 k_1}{\alpha} } \frac{a - \Delta^2 l_2^2}{b - 2 \Delta l_2} \sech \bigg( \sqrt{-k_1 \frac{a - \Delta^2 l_2^2}{b - 2 \Delta l_2} } t \bigg) \tanh \bigg( \sqrt{-k_1 \frac{a - \Delta^2 l_2^2}{b - 2 \Delta l_2} } t \bigg) \,.
\end{cases} 
\eeq

\subsection{Time-Independent External Force}\label{subsec_Duffing1}
Let us consider the system \eqref{motion_k20}, i.e., the system subject to the time-independent perturbation $\textbf{F}$. In this case, once $M_1(\Delta)$ is computed, one might exploit the dependence of $M_1$ on the scalar parameter $\Delta$ and apply Theorem 4.2 of \cite{Wiggins}, that generalizes the Melnikov method for autonomous perturbed Hamiltonian systems. By the way, we can show that the first non-vanishing Melnikov integral appears at least at the second-order in the series \eqref{pert_series}, under proper assumptions; furthermore, we do not need to assume the constraint \eqref{constraint_delta}.
\begin{prop}\label{prop_M_indep}
Let the perturbed Hamiltonian system \eqref{motion_k20} be defined by the Hamiltonian \eqref{Hamiltonian}, the parameters \eqref{parameters}, and the perturbations \eqref{perturbations}. Let $M(\Delta)$, $M_1(\Delta)$ be the zero and the first-order Melnikov integrals of the system \eqref{motion_k20}, respectively, and let $( \varphi_1(t), v_1(t) )$ be the parametrization of the separatrix \eqref{separatrix_solved}. Therefore, under Assumptions \ref{as1} and \ref{as2}:
\begin{itemize}
\item $M(\Delta) = 0$ for every choice on the parameters;
\item $M_1(\Delta) = 0$, provided that $l_2$ is sufficiently close to a certain $L_2$ and for every choice on the remaining parameters.
\end{itemize} 
\end{prop}
\begin{proof}
The first part of the thesis follows immediately from the analytical form of the perturbations $g_{1,2}(\varphi_1)$, which are odd with respect to $\varphi_1$, so they have the same parity of $\varphi_1$ with respect to $t$; furthermore, if $\varphi_1(t)$ has definite parity with respect to $t$, i.e., $\varphi_1(-t) = \pm \varphi_1(t)$, the momentum $v_1(t)$ has opposite parity. Given that, if $\varphi_1(t)$ is the parametrization of the separatrix and has definite parity with respect to $t$, the integrand function in \eqref{Melnikov_integral} is always odd; therefore, $M(\Delta) = 0$ for every choice on the parameters. \\For what concerns $M_1(\Delta)$, let us explicitly rewrite the first subsystem in \eqref{motion_k20}; from Assumption \ref{as1}, i.e., $v_{20} = 0$, we obtain
\beq\label{motion_k20_indep}
\begin{cases}
\dot{\varphi_1} = v_1 \\
\dot{v_1} = -k_1 \varphi_1 - \frac{v_1^2}{2 A^2} \bigg( A \frac{\partial A_{22}}{\partial \varphi_1} - A_{22} \frac{\partial A}{\partial \varphi_1} \bigg) + F g_1(\varphi_1) - F \Delta g_2(\varphi_1) \,.
\end{cases} 
\eeq
After a first-order Taylor expansion on $\varphi_1$, we obtain
\beq\label{motion_k20_indep_as12}
\begin{cases}
\dot{\varphi_1} = v_1 \\
\dot{v_1} = -k_1 \varphi_1 - \frac{1}{4 \Delta l_2} ( a - b \Delta l_2 + \Delta^2 l_2^2 ) \frac{v_1^2}{\varphi_1} + F l_2 \frac{A_{11} - \Delta}{a - \Delta^2 l_2^2} \varphi_1 \,,
\end{cases} 
\eeq
where we assume the definitions \eqref{constraint_parameters}. We now use the canonical transformation
\beq\label{canonical_transf}
\begin{cases}
\tilde{\varphi}_1 = \sqrt{k_1} \varphi_1 \\
\tilde{v}_1 = \frac{1}{ \sqrt{k_1} } v_1 \,,
\end{cases} 
\eeq
so that the \eqref{motion_k20_indep} can be written in the form
\beq\label{eq_Chen}
\begin{cases}
\dot{\tilde{\varphi}}_1 = k_1 \tilde{v}_1 \\
\dot{\tilde{v}}_1 = -\tilde{\varphi}_1 + f(\tilde{\varphi}_1) + \frac{1}{l_2} P_1(\tilde{\varphi}_1, \tilde{v}_1) + F P_2(\tilde{\varphi}_1) \,;
\end{cases} 
\eeq
in particular, we take $f(\tilde{\varphi}_1) = 0$ and
\beq\begin{split}
&P_1(\tilde{\varphi}_1, \tilde{v}_1) = -\frac{\sqrt{k_1}}{4 \Delta} ( a - b \Delta l_2 + \Delta^2 l_2^2 ) \frac{ \tilde{v}_1^2}{ \tilde{\varphi}_1 } \\
&P_2(\tilde{\varphi}_1) = \frac{1}{\sqrt{k_1}} \frac{ l_2 (A_{11} - \Delta)}{a - \Delta^2 l_2^2} \tilde{\varphi}_1 \,.
\end{split}\eeq
The term $P_2(\tilde{\varphi}_1)$ can be seen as perturbative, modulated by the parameter $F$. In order to treat $P_1(\tilde{\varphi}_1, \tilde{v}_1)$ as a perturbative term, we make sure that $\frac{a - b \Delta l_2 + \Delta^2 l_2^2 }{l_2}$ is sufficiently small, i.e., we consider values of $l_2$ in a neighborhood of
\beq
L_2 = \frac{ -( \Delta^2 + A_{11} M ) + \sqrt{\Delta^2 (1 - 4 A_{11} M) + A_{11} M } }{ 2 \Delta } \,.
\eeq
Now, since $M(\Delta) = 0$ and the two perturbations are time-independent, the first-order Melnikov integral reads \cite{Chen}
\beq\label{M1}
\begin{split}
M_1(\Delta) =& \int_{-\infty}^{+\infty} \int_{0}^{\tau_2}  \tilde{v}_1(\tau_2') P_1(\tilde{\varphi}_1(\tau_2'), \tilde{v}_1(\tau_2') ) \frac{ \partial P_1(\tilde{\varphi}_1(\tau_1), \tilde{v}_1(\tau_1) )  }{\partial \tilde{v}_1} d\tau_1 d\tau_2' + \\
&+ \int_{-\infty}^{+\infty} \int_{0}^{\tau_2} \tilde{v}_1(\tau_2') P_2(\tilde{\varphi}_1(\tau_2') ) \frac{ \partial P_1(\tilde{\varphi}_1(\tau_1), \tilde{v}_1(\tau_1) )  }{\partial \tilde{v}_1} d\tau_1 d\tau_2' \,,
\end{split}
\eeq
where $\tilde{\varphi}_1$, $\tilde{v}_1$ satisfy \eqref{separatrix_solved} and \eqref{canonical_transf}. \\Using the Fubini theorem, we can exchange the order of integration; finally, since the term $\frac{ \partial P_1}{\partial \tilde{v}_1}$ is odd with respect to $\tau_1$, we conclude that $M_1(\Delta) = 0$ for every choice on the other parameters.
\end{proof}
Before going on, let us remark upon a few points in the previous proof. Even if \eqref{eq_Chen} presents a perturbative form different from \eqref{motion_k20}, we are allowed to keep taking the homoclinic solution \eqref{separatrix_solved}; indeed, in \eqref{eq_Chen} we take a part of the Hamiltonian as perturbative, so that the relationship \eqref{separatrix_solved} holds also in this case. Furthermore, it is not difficult to see that the minus sign in \eqref{eq_Chen}, i.e., the only difference between \eqref{eq_Chen} and the system studied in \cite{Chen} does not modify the expression of $M_1(\Delta)$, if one takes proper modifications in the proof of Theorem 1 in \cite{Chen}.

Given Proposition \ref{prop_M_indep}, we conclude that the Melnikov method is not applicable at the zero and at the first-order. In order to apply it, we need to compute the second-order Melnikov function, whose expression involves multiple integrals in three variables and several terms; furthermore, to our knowledge, an explicit formula for its computation has not been found, such as in \cite{Chen}, so we ignore it for the purpose of this work.

\subsection{Time-Periodic External Force}\label{subsec_Duffing2}
We now focus on an external force $\textbf{F}$ with time-periodic magnitude, i.e., $\textbf{F} = F \cos(\omega t) \hat F$, so that the perturbed system \eqref{motion_k20} becomes
\beq\label{motion_k20_periodic}
\begin{cases}
\dot{\varphi_1} = v_1 \\
\dot{v_1} = -\frac{\partial H}{\partial \varphi_1} (\varphi_1, v_1; 0) + F \cos(\omega t) g_1(\varphi_1) - F \cos (\omega t) \Delta g_2(\varphi_1) \\
\dot{\varphi_2} = 0 \\
\dot{v_2} = - F \cos(\omega t) g_1(\varphi_1) \,,
\end{cases} 
\eeq
where $g_{1,2}(\varphi_1)$ are defined as in \eqref{perturbations}; in this case, the Melnikov method can be applied in its original formulation. By taking a Taylor expansion of the perturbative terms $g_{1,2}(\varphi_1)$ up to the second order on $\varphi_1$, we arrive to the following expression for $M(t_0; \Delta)$:
\beq\begin{split}
M(t_0; \Delta) = \frac{(\Delta - A_{11}) l_2}{a - \Delta^2 l_2^2} \int_{-\infty}^{+\infty} v_1(t) \frac{\varphi_1(t)}{1 + \frac{\Delta^2 l_2^2}{a - \Delta^2 l_2^2} \varphi_1^2(t)} \cos( \omega(t+t_0) ) dt \,.
\end{split}\eeq
After the substitution $t' = \sqrt{-k_1 \frac{a - \Delta^2 l_2^2}{b - 2 \Delta l_2} } t + y$, we obtain
\beq\label{Melnikov_integral_explicit}
M(t_0'; \Delta) = \frac{A_{11} - \Delta}{b - 2 \Delta l_2} \sqrt{\frac{2 k_1}{\alpha}} \int_{-\infty}^{+\infty} \frac{ \sinh(t') \cos( \omega'(t' + t_0') ) }{ \cosh(t') \bigg( \cosh^2(t') + \beta \bigg) } dt' \,,
\eeq
where
\beq\begin{split}\label{parameters2}
&\omega' = \frac{\omega}{ \sqrt{-k_1 \frac{a - \Delta^2 l_2^2}{b - 2 \Delta l_2} } } \\
&t_0' = \sqrt{-k_1 \frac{a - \Delta^2 l_2^2}{b - 2 \Delta l_2} } t_0 \\
&\beta = \frac{ 2 \Delta^2 l_2^2}{\alpha (a - \Delta^2 l_2^2) } \bigg| \frac{a - \Delta^2 l_2^2}{b - 2 \Delta l_2} \bigg| = \frac{\Delta l_2 (2 \Delta l_2 - b)}{b - 2 \Delta l_2 - a + \Delta^2 l_2^2} \,.
\end{split}\eeq

The definite integral in \eqref{Melnikov_integral_explicit}, solvable by the theorem of residues, can be computed by integrating the function along a rectangular path $\Gamma$ in the complex plane, considering ${\rm Re}(z) \in (-R, R)$, ${\rm Im}(z) \in [0, \pi]$ and taking the limit $R \to +\infty$. We end up with
\beq\label{integral_beta}
\int_{-\infty}^{+\infty} \frac{ \sinh(t') \cos( \omega'(t' + t_0') ) }{ \cosh(t') \bigg( \cosh^2(t') + \beta \bigg) } dt' = -\frac{2 \pi}{1 - e^{-\pi \omega'}} \sin(\omega' t_0') \mathcal{R}(\beta) \,,
\eeq
where $\mathcal{R}(\beta)$ is the sum of the residues that must be taken into account, depending on the value of $\beta$. In order to simplify the notation, we put $t_\pm := \cosh^{-1} (\pm \sqrt{-\beta} )$; furthermore, we set $\beta > -1$ in order to have a convergent integral.
\begin{itemize}
\item $\beta > 0$ \\The integrated function has infinitely many poles of first order in $t' = \frac{i \pi}{2} + 2 k \pi i$, $k \in \mathbb{Z}$ and two poles of first order in $t' = t_\pm$; in this case $t_\pm = \ln(\pm \sqrt{\beta} + \sqrt{\beta + 1}) + \frac{i \pi}{2}$, so that both poles lie inside the path. The residue is
\beq\label{residue1}
\mathcal{R} = \frac{1}{\beta} e^{-\pi \omega'} - \frac{1}{\beta} e^{-\frac{\pi \omega'}{2}} \cos \big( \omega' \ln(\sqrt{\beta} + \sqrt{\beta + 1}) \big) \,,
\eeq
where we used $\rm{Re}(t_+) = -\rm{Re}(t_-)$.
\item $\beta = 0$ \\The integrand function has infinitely many poles of third order in $t' = \frac{i \pi}{2} + 2 k \pi i$, $k \in \mathbb{Z}$. The residue is
\beq\label{residue2}
\mathcal{R} = e^{- \frac{\pi \omega'}{2} } \bigg( \frac{3}{4} - \omega'^2 \bigg) \,.
\eeq
\item $-1 < \beta < 0$ \\The integrand function has infinitely many poles of first order in $t' = \frac{i \pi}{2} + 2 k \pi i$, $k \in \mathbb{Z}$ and two poles of first order in $t' = t_\pm$; in this case $t_\pm = i \tan^{-1} \big( \frac{\sqrt{\beta + 1}}{\pm \sqrt{-\beta}} \big) $, so that only $t_+$ lies inside the path. The residue is
\beq\label{residue3}
\mathcal{R} =  \frac{1}{\beta} e^{-\pi \omega'} - \frac{1}{2 \beta} e^{-\omega' \tan^{-1} \big( \frac{\sqrt{\beta + 1}}{\sqrt{-\beta}} \big) } \,.
\eeq
\end{itemize}
Notice that $\beta \ge 0$ does not satisfy the constraint \eqref{constraint_delta}, so that the only meaningful residue is defined in \eqref{residue3}.

We collect the previous results in the following 
\begin{prop}\label{prop_M}
Let the perturbed Hamiltonian system \eqref{motion_k20_periodic} be defined by the Hamiltonian \eqref{Hamiltonian}, the parameters \eqref{parameters} and the perturbations \eqref{perturbations}. Then, for $v_{20} = 0$ and sufficiently small initial conditions $\varphi_{10}$, the Melnikov integral of the system is
\[
M(t_0'; \Delta) = 2 \pi \frac{ \Delta - A_{11} }{\Delta} l_2^2 \sqrt{ \frac{k_1 \beta}{ 2 \Delta l_2 - b } } \frac{ \mathcal{R}(\beta) }{1 - e^{-\pi \omega'}} \sin(\omega' t_0') \,,
\]
where $a$, $b$ are defined in \eqref{constraint_parameters} and satisfy the constraint \eqref{constraint_delta}, with $\omega'$, $t_0'$, $\beta$ defined in \eqref{parameters2} and
\[
\mathcal{R}(\beta) =  \frac{1}{\beta} e^{-\pi \omega'} - \frac{1}{2 \beta} e^{-\omega' \tan^{-1} \big( \frac{\sqrt{\beta + 1}}{\sqrt{-\beta}} \big) } \,.
\]
\end{prop}
Despite its complicated analytical form, $M(t_0'; \Delta)$ depends on $t_0'$ just through the term $\sin(\omega' t_0')$; therefore, $M(t_0'; \Delta)$ has simple zeros in $t_0' = \frac{k \pi}{\omega'}$, $k \in \mathbb{Z}$, corresponding to transverse homoclinic intersections. Moreover, since
\beq
\frac{dM(t_0'; \Delta)}{dt_0'} = 2 \pi \frac{ \Delta - A_{11} }{\Delta} l_2^2 \sqrt{ \frac{k_1 \beta}{ 2 \Delta l_2 - b } } \frac{ \mathcal{R}(\beta) }{1 - e^{-\pi \omega'}} \omega' \cos(\omega' t_0') \,,
\eeq
$M(t_0'; \Delta)$ has not double zeros.

We thus end up with the conclusion that the Melnikov method predicts chaotic behavior for every choice on the parameters that satisfy \eqref{constraint_delta}, excluding the singular cases corresponding to $A_{11} = \Delta$, $A_{11} + M l_2^2 = 2 \Delta l_2$ and $\mathcal{R}(\beta) = 0$. This is not unexpected since we impose from the beginning that the system behaves like a Duffing oscillator that is known to be chaotic if not subject to dissipation \cite{Cencini}; however, we remark that we found the specific constraint \eqref{constraint_delta} on the parameters of the system such that this is actually the case. Some numerical examples are shown in Figures \ref{periodic_res3_omega1}--\ref{periodic_res3_omega100}.
\begin{figure}[H]
\includegraphics[width=12.5cm]{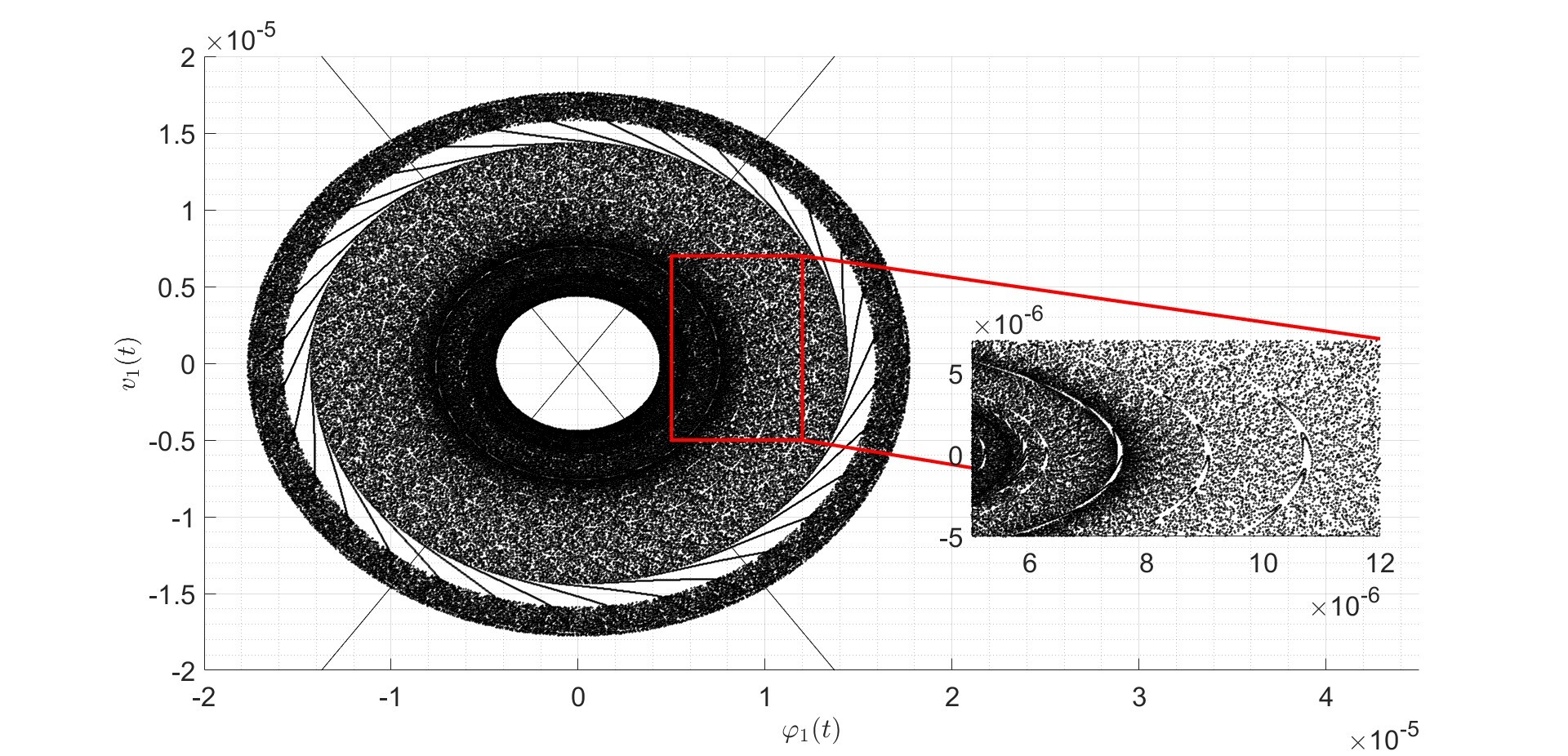}
\caption{{Motion 
 of the system \eqref{motion_k20_periodic} around the separatrix \eqref{separatrix_solved}. Initial conditions: $\varphi_1(0) = 10^{-5}$, $v_1(0) = 1.4577 \cdot 10^{-5}$. Parameters: $k_1 = 1$, $l_2 = 1$, $A_{11} = 1$, $M = 4$, $\Delta = 2.25$, $F = 0.01$, $\omega = 1$.} }
\label{periodic_res3_omega1}
\end{figure}
\unskip
\begin{figure}[H]
\includegraphics[width=12.5cm]{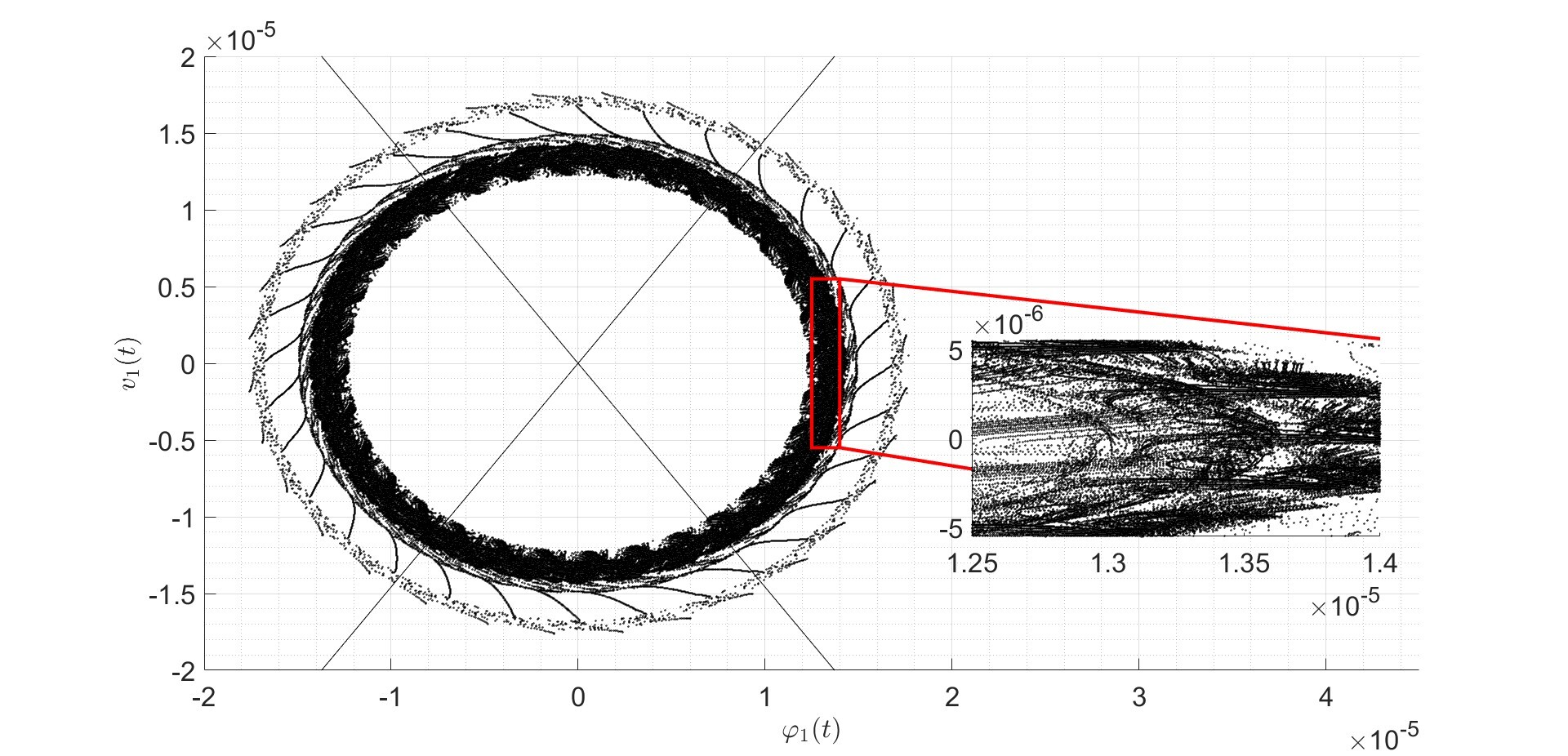}
\caption{{Motion
 of the system \eqref{motion_k20_periodic} around the separatrix \eqref{separatrix_solved}. Initial conditions: $\varphi_1(0) = 10^{-5}$, $v_1(0) = 1.4577 \cdot 10^{-5}$. Parameters: $k_1 = 1$, $l_2 = 1$, $A_{11} = 1$, $M = 4$, $\Delta = 2.25$, $F = 0.01$, $\omega = 10$.} }
\label{periodic_res3_omega10}
\end{figure}
\unskip
\begin{figure}[H]
\includegraphics[width=12.5cm]{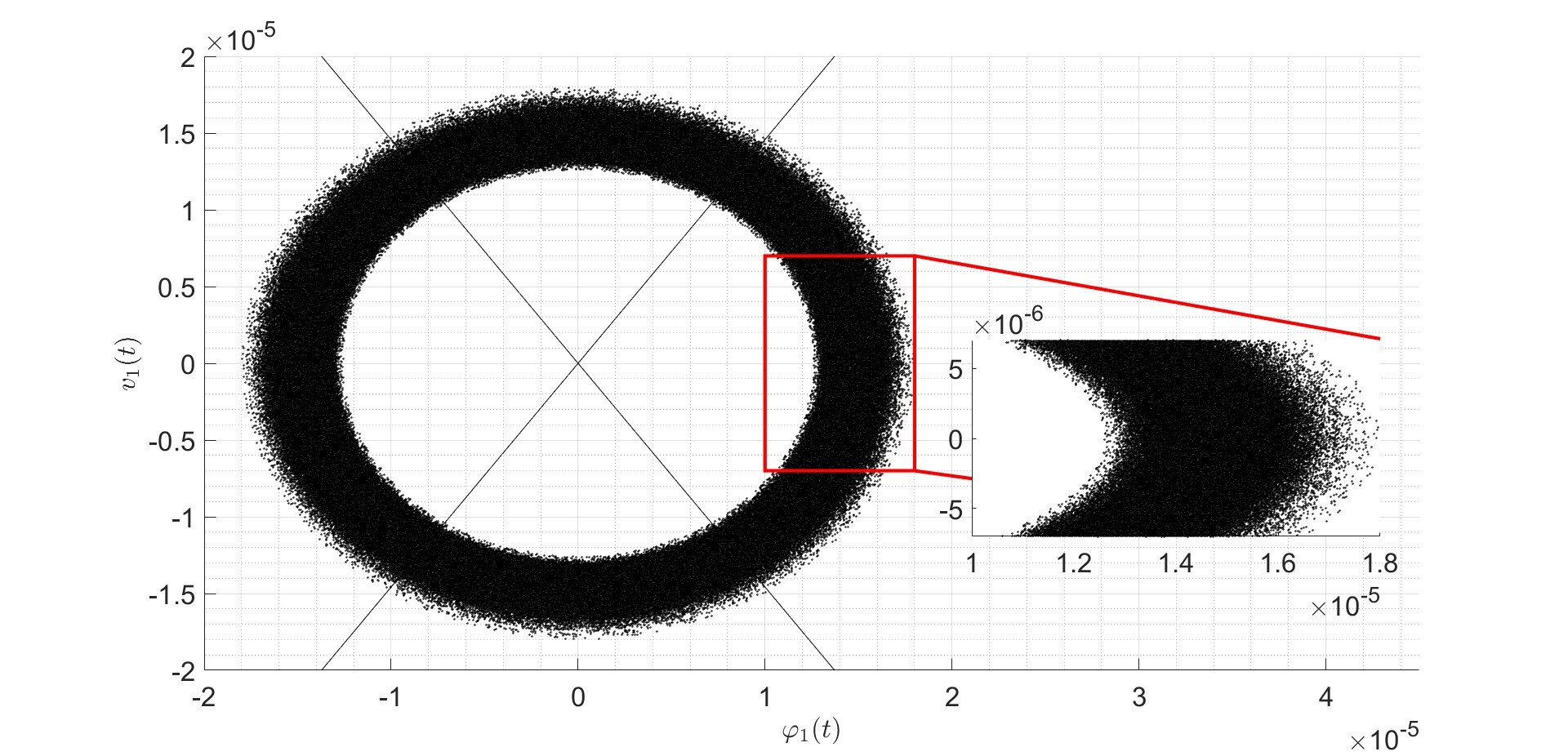}
\caption{{Motion
 of the system \eqref{motion_k20_periodic} around the separatrix \eqref{separatrix_solved}. Initial conditions: $\varphi_1(0) = 10^{-5}$, $v_1(0) = 1.4577 \cdot 10^{-5}$. Parameters: $k_1 = 1$, $l_2 = 1$, $A_{11} = 1$, $M = 4$, $\Delta = 2.25$, $F = 0.01$, $\omega = 100$.} }
\label{periodic_res3_omega100}
\end{figure}
\unskip

\subsection{Time-Periodic External Force and Dissipative Term}\label{subsec_Duffing3}
Let us now analyze the problem in the presence of a dissipative term acting on the lower rod of the pendulum. The Equation \eqref{motion_k20} become
\beq\label{motion_k20_nu}
\begin{cases}
\dot{\varphi_1} = v_1 \\
\dot{v_1} = -\frac{\partial H}{\partial \varphi_1} (\varphi_1, v_1; 0) + F \cos(\omega t) g_1(\varphi_1) - F \cos(\omega t) \Delta g_2(\varphi_1) - \nu \dot{\varphi}_1 \\
\dot{\varphi_2} = 0 \\
\dot{v_2} = - F \cos(\omega t) g_1(\varphi_1) \,,
\end{cases} 
\eeq
where $F$, $\nu$ are perturbative parameters. In this case, the Melnikov integral is given by the sum of $M(t_0; \Delta)$ defined in Proposition \ref{prop_M} and the following contribution due to the dissipative term:
\beq\label{Melnikov_nu}
\int_{-\infty}^{+\infty} v_1^2(t) dt = \frac{2 \sqrt{k_1}}{\alpha} \bigg( \frac{a - \Delta^2 l_2^2 }{2 \Delta l_2 - b} \bigg)^\frac{3}{2} \int_{-\infty}^{+\infty} \frac{ \sinh^2(t') }{ \cosh^4(t') } dt' \,,
\eeq
where the definitions and substitutions taken above are understood. The integrand function has infinitely many poles of fourth order in $t' = \frac{i \pi}{2} + 2k \pi i$, $k \in \mathbb{Z}$; the integral can be solved using the theorem of residues by integrating along the same path considered for the computation of $M(t_0'; \Delta)$.
\begin{prop}\label{prop_Mnu}
Let the perturbed Hamiltonian system \eqref{motion_k20_nu} be defined by the Hamiltonian \eqref{Hamiltonian}, the parameters \eqref{parameters}, and the perturbations \eqref{perturbations}. Then, for $v_{20} = 0$ and sufficiently small initial conditions $\varphi_{10}$, the Melnikov integral of the system is
\[
M_\nu(t_0'; \Delta) = M(t_0'; \Delta) + \frac{2}{3 \Delta l_2 (b - 2 \Delta l_2 + a - \Delta^2 l_2^2) } \sqrt{ k_1 \frac{(a - \Delta^2 l_2^2)^3}{(2 \Delta l_2 - b)} } \,,
\]
where $M(t_0'; \Delta)$ is defined in Proposition \ref{prop_M} and $a$, $b$ are defined in \eqref{constraint_parameters}, satisfying the constraint \eqref{constraint_delta}.
\end{prop}
Given the expression stated in Proposition \ref{prop_Mnu}, if we define
\beq\begin{split}
&A(\Delta) = 2 \pi \frac{ \Delta - A_{11} }{\Delta} l_2^2 \sqrt{ \frac{k_1 \beta}{ 2 \Delta l_2 - b } } \frac{ \mathcal{R}(\beta) }{1 - e^{-\pi \omega'}} \\
&B(\Delta) = \frac{2}{3 \Delta l_2 (b - 2 \Delta l_2 + a - \Delta^2 l_2^2) } \sqrt{ k_1 \frac{(a - \Delta^2 l_2^2)^3}{(2 \Delta l_2 - b)} } \,,
\end{split}\eeq
we can write $M(t_0'; \Delta)$ in the form
\beq\begin{split}
&M_\nu(t_0'; \Delta)  = A(\Delta) \sin(\omega' t_0') + B(\Delta) \\
&\frac{dM_\nu(t_0'; \Delta)}{dt_0'} = A(\Delta) \omega' \cos(\omega' t_0')
\end{split}\eeq
and we obtain of course $-A(\Delta) + B(\Delta) \le M_\nu(t_0'; \Delta) \le A(\Delta) + B(\Delta)$. Therefore, a threshold for $M_\nu(t_0'; \Delta)$ is given by the sign of $A(\Delta) - B(\Delta)$: $M_\nu(t_0'; \Delta)$ has simple zeros if $A(\Delta) > B(\Delta)$, double zeros if $A(\Delta) = B(\Delta)$ and no zeros if $A(\Delta) < B(\Delta)$. Finally, the Melnikov method predicts chaotic behavior for the system \eqref{motion_k20_nu} if
\beq\label{result_k20_nu}
3 \pi \Delta ( \Delta - A_{11} ) (2 \Delta l_2 - b) l_2^4 \frac{\mathcal{R}(\beta)}{1 - e^{-\pi \omega'}} \le \sqrt{ \beta (a - \Delta^2 l_2^2)^3} \,.
\eeq
A numerical example is shown in Figure \ref{periodic_res3_nu}.
\begin{figure}[H]
\includegraphics[width=12.5cm]{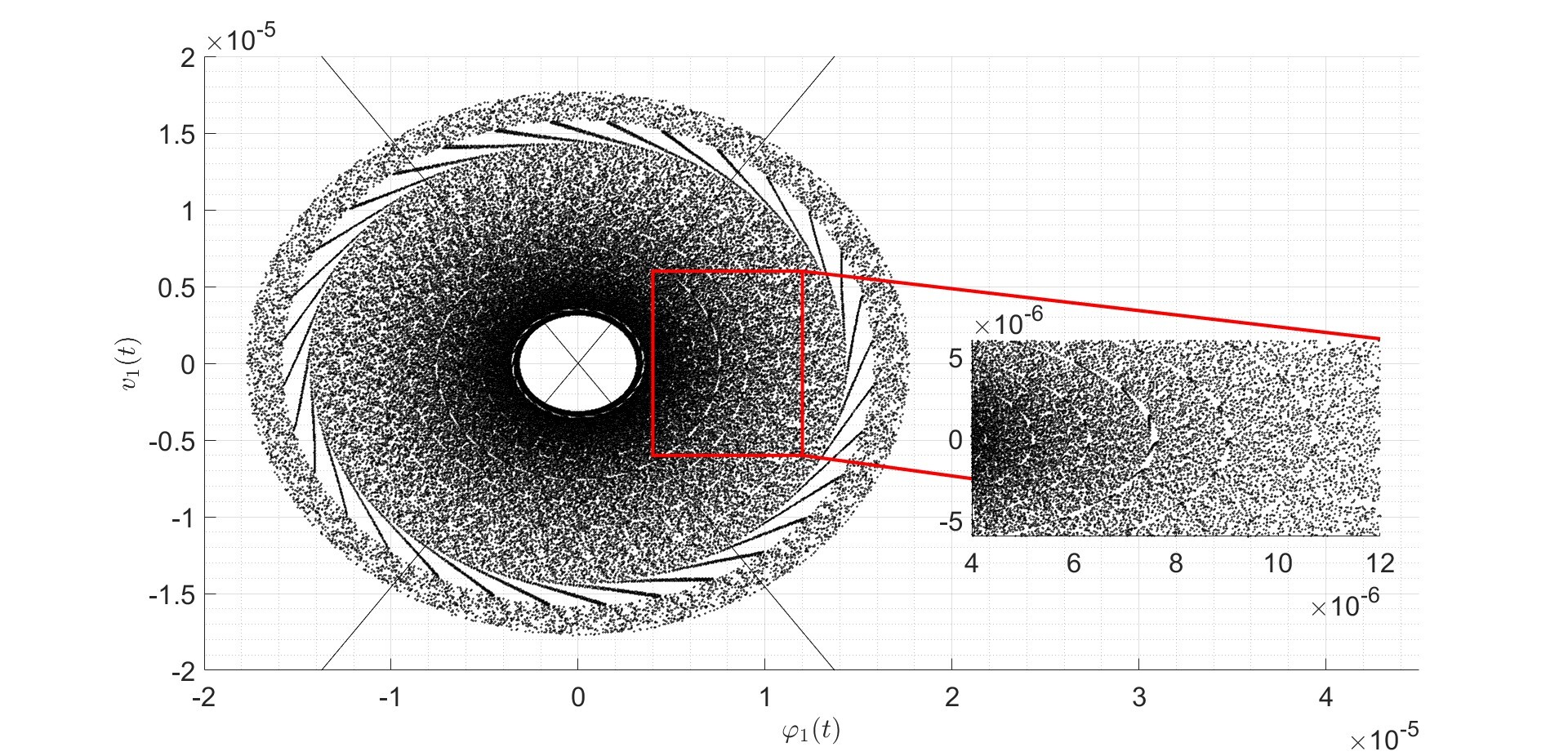}
\caption{{Motion
 of the system \eqref{motion_k20_nu} around the separatrix \eqref{separatrix_solved}. Initial conditions: $\varphi_1(0) = 10^{-5}$, $v_1(0) = 1.4577 \cdot 10^{-5}$. Parameters: $k_1 = 1$, $l_2 = 1$, $A_{11} = 1$, $M = 4$, $\Delta = 2.25$, $F = 0.01$, $\omega = 0.3$, $\nu = 10^{-4}$.} }
\label{periodic_res3_nu}
\end{figure}
\unskip

\section{Melnikov Integral for $\Delta = 0$}\label{sec_Delta0}
The case $\Delta = 0$ is of particular interest since for this choice of the parameters, the system is integrable in a weak sense, and periodic solutions arise \cite{Polekhin}; we study this case separately since $\Delta = 0$ does not satisfy the constraint \eqref{constraint_delta}.
\begin{prop}\label{prop_delta0_M_indep}
Let the perturbed Hamiltonian system \eqref{motion_k20} be defined by the Hamiltonian \eqref{Hamiltonian}, the parameters \eqref{parameters}, and the perturbations \eqref{perturbations}. Then, if $\Delta = 0$, the Melnikov integral of the system does not exist.
\end{prop}
\begin{proof}
For $\Delta = 0$, \eqref{motion_k20} becomes
\beq\label{motion_delta0}
\begin{cases}
\dot{\varphi_1} = v_1 \\
\dot{v_1} = -k_1 \varphi_1 + \frac{F}{M l_2} \sin(\varphi_1) \\
\dot{\varphi_2} = v_{20} \\
\dot{v_2} = - \frac{F}{M l_2} \sin(\varphi_1) \,.
\end{cases} 
\eeq
The Melnikov integral $M(\Delta)$ is
\beq\label{Melnikov_integral_delta0}
M(0) = -\frac{1}{M l_2} \int_{-\infty}^{+\infty} \dot{\varphi_1}(t) \sin(\varphi_1(t)) dt = \frac{1}{M l_2} \cos(\varphi_1(t)) \bigg|_{-\infty}^{+\infty} \,.
\eeq
Since, in this case, the system has a family of periodic solutions, there exists a time $T$ such that $(\varphi_1(t+T), v_1(t+T)) = (\varphi_1(t), v_1(t))$ $\forall t$; therefore, $\nexists M(0)$.
\end{proof}
Given Proposition \ref{prop_delta0_M_indep}, we are not allowed to apply the Melnikov method in any order for this specific choice of the parameters; indeed, in this case, we end up with a totally undefined Melnikov function, so we cannot define a priori the perturbative series \eqref{pert_series}.

Now, we consider, as before, the presence of a time-periodic external force. In order to apply the Melnikov method, we write the external force $\textbf{F}$ as the sum of a constant term plus a small time-periodic perturbation, i.e., we redefine $\textbf{F}$ as
\beq\label{F_sum}
F = F_0 + F_\varepsilon \cos(\omega t) \,.
\eeq
Therefore, we can write the \eqref{motion} for $\Delta = 0$ in the following perturbative form:
\beq\label{motion_delta0_periodic}
\begin{cases}
\dot{\varphi_1} = v_1 \\
\dot{v_1} = \frac{F_0}{M l_2} \sin(\varphi_1) + k_1 g_3(\varphi_1) + F_\varepsilon \cos(\omega t) g_4(\varphi_1) \\
\dot{\varphi_2} = v_{20} \\
\dot{v_2} = - \frac{F_0}{M l_2} \sin(\varphi_1) - F_\varepsilon \cos(\omega t) g_4(\varphi_1) \,,
\end{cases}
\eeq
where
\beq\label{perturbations_delta0}
g_3(\varphi_1) = -\varphi_1, \quad g_4(\varphi_1) = \frac{1}{M l_2} \sin(\varphi_1) \,,
\eeq
$k_1$, $F_\varepsilon$ are perturbative parameters and 
\beq\label{omega0}
\omega_0^2 := \frac{F_0}{M l_2} \,.
\eeq
The system \eqref{motion_delta0_periodic} represents the equations of motion of a simple pendulum with natural frequency $\omega_0$, subject to time-independent and time-periodic perturbations. The analytical form for the separatrix of the unperturbed system is well-known \cite{Trueba}:
\beq\label{separatrix_delta0_solved}
\begin{cases}
\varphi_1(t) = \pm 2 \tan^{-1}(\sinh(\omega_0 t)) + \pi \\
v_1(t) = \pm 2 \omega_0 \sech(\omega_0 t) \,.
\end{cases} 
\eeq

From Proposition \ref{prop_M_indep}, we have a vanishing Melnikov integral associated with the time-independent perturbation. Therefore, the only contribution to the Melnikov integral is given by the time-periodic perturbation; after a Taylor expansion on $\varphi_1(t)$, we obtain
\beq\label{Melnikov_integral_delta0_periodic}
M(t_0; 0) = \frac{2 \omega_0}{M l_2} \int_{-\infty}^{+\infty} \sech(\omega_0 t) \sin\big( 2 \sech(\omega_0 t) \tanh(\omega_0 t) \big) \cos(\omega (t + t_0) ) dt {\,.}
\eeq
The expression of $M(t_0; \Delta)$ follows directly from \eqref{integral_beta} and \eqref{residue2}; notice that in this case we do not need to take Assumption \ref{as1}, i.e., we assume in general $v_{20} \ne 0$.
\begin{prop}\label{prop_delta0_M}
Let the perturbed Hamiltonian system \eqref{motion_delta0_periodic} be defined by the Hamiltonian \eqref{Hamiltonian}, the parameters \eqref{parameters}, the perturbations \eqref{perturbations_delta0}, the external force \eqref{F_sum} and the natural frequency \eqref{omega0}. Then, for sufficiently small initial conditions $\varphi_{10}$ the Melnikov integral of the system is
\[
M(t_0; 0) = -\frac{8 \pi \omega_0}{M l_2} \frac{ e^{ -\frac{\pi \omega}{2 \omega_0} } }{ 1 - e^{ -\frac{\pi \omega}{\omega_0} } } \bigg( \frac{3}{4} - \frac{\omega^2}{\omega_0^2} \bigg) \sin(\omega t_0) \,.
\]
\end{prop}
From Proposition \ref{prop_delta0_M}, we have that $M(t_0; 0)$ has simple zeros in $t_0 = \frac{k \pi}{\omega_0}$, $k \in \mathbb{Z}$; furthermore, $M(t_0; 0)$ vanishes identically if $\omega = \frac{\sqrt{3}}{2} \omega_0$. Therefore, the Melnikov method predicts chaotic behavior for the system \eqref{motion_delta0_periodic} for every choice on the parameters, excluding the singular case $\omega = \frac{\sqrt{3}}{2} \omega_0$.

In the presence of a dissipative term acting on the lower rod, \eqref{motion} become
\beq\label{motion_delta0_nu}
\begin{cases}
\dot{\varphi_1} = v_1 \\
\dot{v_1} = \frac{F_0}{M l_2} \sin(\varphi_1) + k_1 g_3(\varphi_1) + F_\varepsilon \cos(\omega t) g_4(\varphi_1) - \nu \dot{\varphi}_1 \\
\dot{\varphi_2} = v_{20} \\
\dot{v_2} = - \frac{F_0}{M l_2} \sin(\varphi_1) - F_\varepsilon \cos(\omega t) g_4(\varphi_1) \,,
\end{cases} 
\eeq
where $g_{3,4}(\varphi_1)$ are defined in \eqref{perturbations_delta0} and $\nu$ is a perturbative parameter; for the system \eqref{motion_delta0_nu} the Melnikov integral is given by
\beq\label{Melnikov_integral_delta0_nu}
M_\nu(t_0; 0) = M(t_0; 0) + 4 \omega_0^2 \int_{-\infty}^{+\infty} \sech^2(\omega_0 t) dt \,,
\eeq
where $M(t_0; 0)$ is defined in \eqref{Melnikov_integral_delta0_periodic}.
\begin{prop}\label{prop_delta0_Mnu}
Let the perturbed Hamiltonian system \eqref{motion_delta0_nu} be defined by the Hamiltonian \eqref{Hamiltonian}, the parameters \eqref{parameters}, the perturbations \eqref{perturbations_delta0}, the external force \eqref{F_sum} and the natural frequency \eqref{omega0}. Then, for sufficiently small initial conditions $\varphi_{10}$ the Melnikov integral of the system is
\[
M_\nu(t_0; 0) = 8 \omega_0 \bigg[ 1 + \frac{\pi}{M l_2} \frac{ e^{ -\frac{\pi \omega}{2 \omega_0} } }{ 1 - e^{ -\frac{\pi \omega}{\omega_0} } } \bigg( \frac{\omega^2}{\omega_0^2} - \frac{3}{4} \bigg) \sin(\omega t_0) \bigg] \,.
\]
\end{prop}
From Proposition \ref{prop_delta0_Mnu}, we have that $M_\nu(t_0; 0)$ has simple zeros if $\omega > \frac{\sqrt{3}}{2} \omega_0$ and double zeros if $\omega = \frac{\sqrt{3}}{2} \omega_0$, while it is always positive if $\omega < \frac{\sqrt{3}}{2} \omega_0$. Therefore, the Melnikov method predicts chaotic behavior for the system \eqref{motion_delta0_nu} if
\beq\label{result_delta0_nu}
\omega \ge \sqrt{ \frac{3 F_0}{4 M l_2} } \,.
\eeq
Some numerical examples for both dissipative and non-dissipative cases are shown in Figures \ref{periodic_delta0} and \ref{periodic_delta0_nu}.
\begin{figure}[H]
\includegraphics[width=10cm]{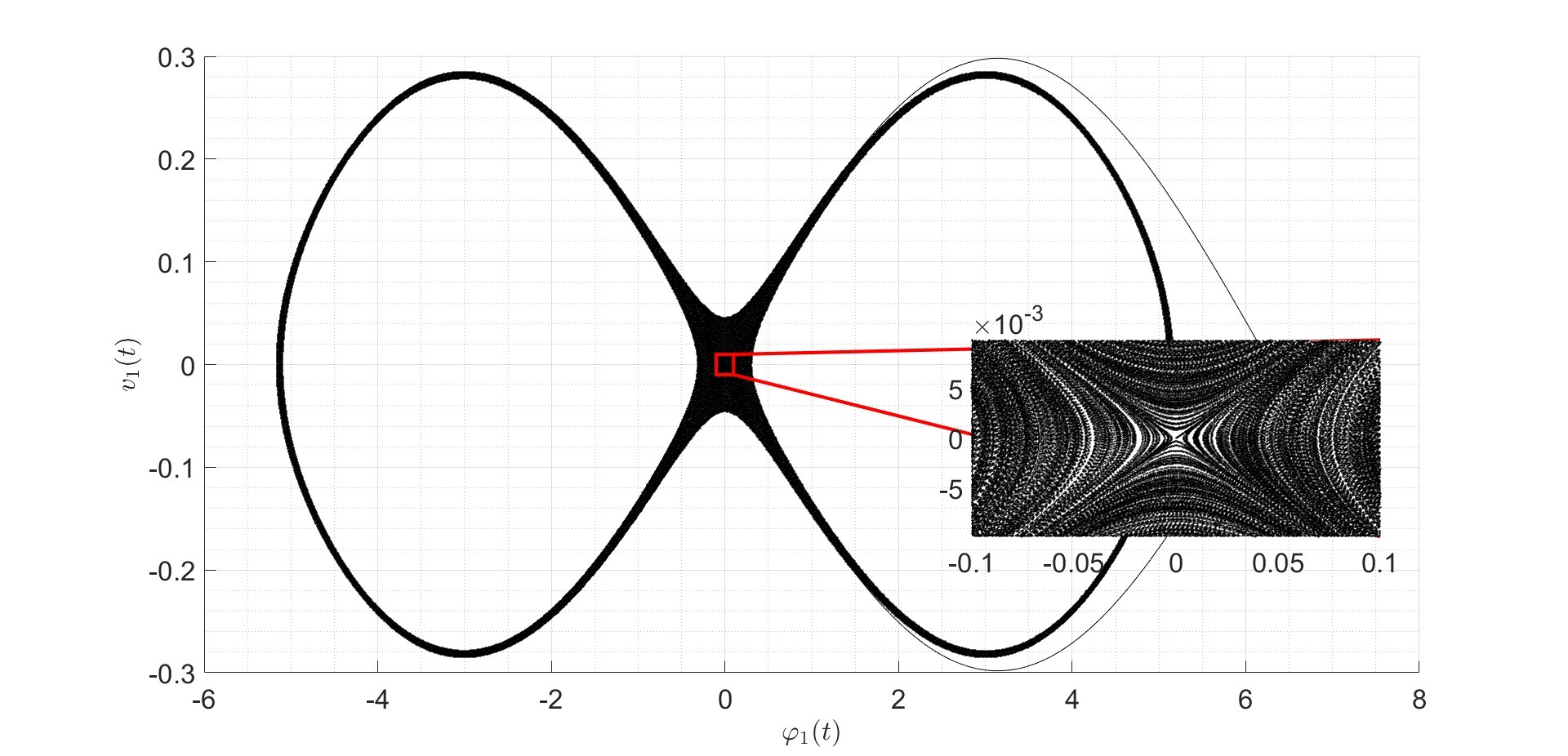}
\caption{{Motion
 of the system \eqref{motion_delta0_periodic} around the separatrix \eqref{separatrix_delta0_solved}. Initial conditions: $\varphi_1(0) = 10^{-5}$, $v_1(0) = 1.4907 \cdot 10^{-6}$. Parameters: $k_1 = 0.001$, $l_2 = 4.5$, $M=10$, $\omega_0 = 0.1491$, $F_\varepsilon = 0.01$, $\omega = 0.5$.} }
\label{periodic_delta0}
\end{figure}
\unskip
\begin{figure}[H]
\includegraphics[width=10cm]{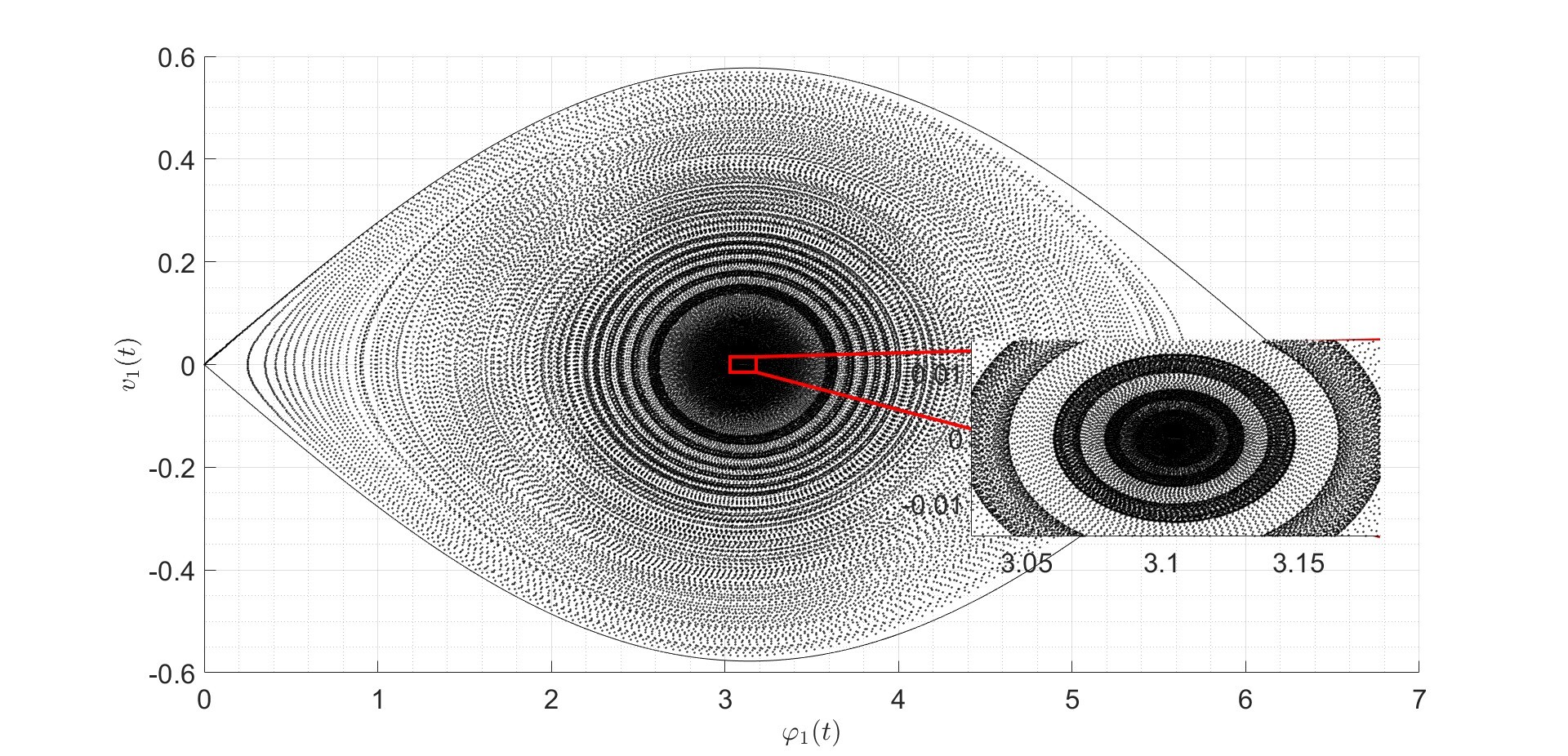}
\caption{{Motion
 of the system \eqref{motion_delta0_nu} around the separatrix \eqref{separatrix_delta0_solved}. Initial conditions: $\varphi_1(0) = 10^{-5}$, $v_1(0) = 2.8868 \cdot 10^{-6}$. Parameters: $k_1 = 0.001$, $l_2 = 1.5$, $M=8$, $\omega_0 = 0.2887$, $F_\varepsilon = 0.01$, $\omega = 2 \omega_0$, $\nu = 5 \cdot 10^{-4}$.} }
\label{periodic_delta0_nu}
\end{figure}
\unskip

\section{Conclusions}\label{sec_conclusions}
In this paper, we applied the Poincaré–Melnikov method to a class of dynamical systems related to the generalized Ziegler pendulums, i.e., a mathematical double-pendulum subject to angular elastic potential and an external follower force. By assuming a few assumptions on the initial conditions, we found the analytical expression for a generic separatrix of the system in terms of elliptic integrals. Under these assumptions and a further constraint on the parameters, we reduced the equations of motion of the system to a Duffing oscillator and computed the Melnikov integral for three possible formulations of the dynamical problem. We showed that the Melnikov method fails at the first order if applied to the original system and fails at all in the case integrable case $\Delta = 0$; in this case, we computed the second-order Melnikov integral and found an explicit threshold on the parameters such that the Melnikov method predicts chaotic motion. Explicit expressions of the first-order Melnikov integral and analogous relationships for the parameters have been found by considering the presence of a time-periodic external force and a dissipation term acting on the lower rod.

The Melnikov method has been applied to several dynamical systems, and the aim of this study is not far from the usual one, i.e., to find suitable relationships in the parameters of the dynamical system in order to predict and control its chaotic behavior.

From a mathematical point of view, some fundamental questions arise.

By considering the system subject to a time-independent external force, i.e., the generalized Ziegler pendulum in its original formulation, we showed that the Melnikov method is not applicable at the zero and the first order. To the best of our knowledge, an explicit expression for the second-order Melnikov integral, such as the first-order one found in \cite{Chen}, is not known in the literature. Therefore, an interesting goal might be a generalization of the results presented in \cite{Chen}, where it has been shown that Melnikov integrals of any order can be derived by an inductive construction; this result would find several applications, including the system we analyzed here.

The three formulations of the dynamical problem we analyzed in this work present choices on the parameters that lead to a vanishing Melnikov integral; in these cases, one must compute the first-order Melnikov integral. Since we encountered several situations of this kind and the computation of the first-order Melnikov integral might be challenging, we leave this problem for future work.

Another question we pose concerns the topological interpretation of an identically null Melnikov integral. Melnikov functions are defined in terms of the splitting distance between stable and unstable manifolds of a Hamiltonian system when subject to a time-periodic perturbation. One can encounter an arbitrary number of identically null terms in the power series before finding a first non-vanishing Melnikov integral, but it is not trivial to interpret the case in which all the terms in the series are vanishing (and there is apparently no reason to exclude this case); if the series \eqref{pert_series} is not simply undefined, an open and challenging question concerns the meaning of having globally coincident stable and unstable manifolds.

Further developments of this work may include applications or generalizations of the analytical results to other pendulum-like systems, such as variable-length pendulums \cite{Szuminski2, Krasilnikov}. A particularly interesting application may concern the swinging Atwood's machine \cite{Tufillaro, Yakubu} and its generalizations \cite{Szuminski3}; since both models include the presence of elastic potentials, one might find relationships between the Melnikov integral of the generalized Ziegler pendulum and the one associated with the SAM system, or try to reduce in same way the equations of motion of one system to the other ones.

We conclude by suggesting examining different assumptions in order to obtain more general results about the dynamics of the system; finally, it might be interesting to have empirical verifications of the theoretical predictions stated in this work.

\backmatter

\bmhead{Acknowledgements} The authors are grateful to the anonymous reviewers for their useful suggestions.

\section*{Declarations}

\bmhead{Author contribution} Conceptualization, S.D. and V. C.; Methodology, S.D. and V. C.; Software, S.D.; Validation, V. C.; Formal analysis, S.D.; Writing--original draft, S.D.; Writing--review and editing, V. C.; Visualization, S.D.; Supervision, V. C.; Project administration, V. C.; Funding acquisition, V. C. All authors have read and agreed to the published version of the manuscript.

\bmhead{Funding} This research was funded by the University of Ferrara, FIRD 2024.

\bmhead{Data Availability Statement} The data that supports the funding of this study are available within the article.

\bmhead{Conflict of Interest} The authors declare no conflicts of interest.

\appendix

\renewcommand{\theequation}{A\arabic{equation}}
\setcounter{equation}{0}
\renewcommand{\theprop}{A\arabic{prop}}
\setcounter{prop}{0}
\renewcommand{\thetheo}{A\arabic{theo}}
\setcounter{theo}{0}

\section[\appendixname \thesection]{Proof of Proposition 4 in \cite{Disca} }
In \cite{Disca}, we defined a discrete map associated with the generalized Ziegler pendulum; assuming the validity of the non-Hamiltonian symmetry $\Delta = 0$, the discrete map is (here we take a redefinition of the coefficients considered in \cite{Disca})
\begin{subequations}\label{map}
\beq\begin{split}
&f \colon \mathbb{R}^4 \to \mathbb{R}^4 \\
&\overline{x}_{n+1} = f(\overline{x}_n), \quad n \in \mathbb{N}
\end{split}\eeq
\beq\label{map_eq}
\begin{split}
&x_{n+1} = y_n \\
&y_{n+1} = \alpha x_n + \beta z_n + \gamma \sin(x_n) \\
&z_{n+1} = \omega_n \\
&\omega_{n+1} = \tilde \alpha x_n - \beta z_n - \gamma \sin(x_n)
\end{split}, \quad \alpha < 0, \tilde \alpha > 0, \beta > 0, \gamma \ne 0 \,.
\eeq
\end{subequations}
By explicit calculations, we proved that the map \eqref{map} does not have dense sets of periodic points up to period 3, and we conjectured that the map has no dense sets of periodic points at all. Since the present work analyzes the chaotic behavior of the same dynamical system studied in \cite{Disca}, it represents an appropriate context for presenting a proof of this latter statement.
\begin{prop}\label{prop_density}
The map \eqref{map} does not have dense sets of periodic points.
\end{prop}
\begin{proof}
The equations \eqref{map_eq} can be written as
\beq\label{map_gh}
\begin{split}
&x_{n+1} = y_n \\
&y_{n+1} = g(x_n) + h(z_n) \\
&z_{n+1} = \omega_n \\
&\omega_{n+1} = \tilde g(x_n) - h(z_n)\,,
\end{split} 
\eeq
where $g(x) = \alpha x + \gamma \sin(x)$, $\tilde g(x) = \tilde \alpha x - \gamma \sin(x)$, $h(z) = \beta z$. We define the two maps
\[
F^{(n)} = F^{(n)} \bigg( g(x_n), \tilde g(x_n), h(z_n) \bigg), \quad G^{(n)} = G^{(n)} \bigg( g(x_n), \tilde g(x_n), h(z_n) \bigg)  \,,
\]
such that $F^{(n)}$ and $G^{(n)}$ collect the compositions of the functions $g$, $\tilde g$ and $h$ up to the $n$-th iteration; for example, starting from an initial point $p_0 = (x_0, y_0, z_0, \omega_0)$, we have the following definitions for the second and the third iterations:
\beq
\begin{split}
&x_2 = g(x_0) + h(z_0) =: F^{(2)}(x_0, z_0) \\
&y_2 = g(y_0) + h(\omega_0) = F^{(2)}(y_0, \omega_0) \\
&z_2 = \tilde g(x_0) - h(z_0) =: G^{(2)}(x_0, z_0) \\
&\omega_2 = \tilde g(y_0) - h(\omega_0) = G^{(2)}(y_0, \omega_0)
\end{split}
\eeq
\beq
\begin{split}
&x_3 = g(y_0) + h(\omega_0) = F^{(2)}(y_0, \omega_0) \\
&y_3 = g \bigg( g(x_0) + h(z_0) \bigg) + h \bigg( \tilde g(x_0) - h(z_0) \bigg) =: F^{(3)}(x_0, \omega_0) \\
&z_3 = \tilde g(y_0) - h(\omega_0) = G^{(2)}(y_0, \omega_0) \\
&\omega_3 = \tilde g \bigg( g(x_0) + h(z_0) \bigg) - h \bigg( \tilde g(x_0) - h(z_0) \bigg) =: G^{(3)}(y_0, \omega_0)\,.
\end{split} 
\eeq
The following structure comes out for a generic iteration of the map:
\beq\label{iter_even}
\begin{split}
&x_{2n} = F^{(2n)}(x_0, z_0) \\
&y_{2n} = F^{(2n)}(y_0, \omega_0) \\
&z_{2n} = G^{(2n)}(x_0, z_0) \\
&\omega_{2n} = G^{(2n)}(y_0, \omega_0)
\end{split} \qquad n \in \mathbb{N}
\eeq
\beq\label{iter_odd}
\begin{split}
&x_{2n+1} = F^{(2n)}(y_0, \omega_0) \\
&y_{2n+1} = g \bigg( F^{(2n)}(x_0, z_0) \bigg) + h \bigg( G^{(2n)}(x_0, z_0) \bigg) \\
&z_{2n+1} = G^{(2n)}(y_0, \omega_0) \\
&\omega_{2n+1} = \tilde g \bigg( F^{(2n)}(x_0, z_0) \bigg) - h \bigg( G^{(2n)}(x_0, z_0) \bigg)
\end{split} \qquad n \in \mathbb{N} \,.
\eeq
Let $p_0 = (x_0, y_0, z_0, \omega_0)$ be a periodic point of period $2n$, $n \in \mathbb{N}$; from \eqref{iter_even}, we obtain
\beq
\begin{split}
&F^{(2n)}(x_0, z_0) = x_0 \\
&F^{(2n)}(y_0, \omega_0) = y_0 \\
&G^{(2n)}(x_0, z_0) = z_0 \\
&G^{(2n)}(y_0, \omega_0) = \omega_0\,.
\end{split} 
\eeq
The previous system can be split into two systems that share the same set of solutions, which are
\beq
\begin{split}
&F^{(2n)}(x_0, z_0) = x_0 \\
&G^{(2n)}(x_0, z_0) = z_0 \\
\end{split} \quad \cap \quad
\begin{split}
&F^{(2n)}(y_0, \omega_0) = y_0 \\
&G^{(2n)}(y_0, \omega_0) = \omega_0
\end{split} 
\eeq
so that a periodic point of even period must be of the form $p_0 = (x_0, x_0, z_0, z_0)$, where $x_0$ and $z_0$ are solutions of
\beq
\begin{split}
&F^{(2n)}(x_0, z_0) = x_0 \\
&G^{(2n)}(x_0, z_0) = z_0\,.
\end{split} 
\eeq
On the other hand, since $y_n = x_{n+1}$ and $\omega_n = z_{n+1}$, it must be
\beq
\begin{split}
&y_{2n} = x_0 = x_{2n+1} = x_1 \\
&\omega_{2n} = z_0 = z_{2k+1} = z_1\,,
\end{split} 
\eeq
that is $p_0$ is a fixed point; therefore, the map \eqref{map_gh} has not periodic points of even period. \\Let $p_0 = (x_0, y_0, z_0, \omega_0)$ be a periodic point of period $2n+1$, $n \in \mathbb{N}$; from the \eqref{iter_odd} we obtain
\beq
\begin{split}
&F^{(2n)}(y_0, \omega_0) = x_0 \\
&g( F^{(2n)}(x_0, z_0) ) + h( G^{(2n)}(x_0, z_0) ) = y_0 \\
&G^{(2n)}(y_0, \omega_0) = z_0 \\
&\tilde g( F^{(2n)}(x_0, z_0) ) - h( G^{(2n)}(x_0, z_0) ) = \omega_0\,,
\end{split} 
\eeq
that returns
\begin{subequations}\label{iter_odd_yomega}
\beq\begin{split}
y_0 =& g \bigg( F^{(2n)} \bigg( F^{(2n)}(y_0, \omega_0), G^{(2n)}(y_0, \omega_0) \bigg) \bigg) + \\
&+ h \bigg( G^{(2n)} \bigg( F^{(2n)}(y_0, \omega_0), G^{(2n)}(y_0, \omega_0) \bigg) \bigg)
\end{split}\eeq
\beq\begin{split}
\omega_0 =& \tilde g \bigg( F^{(2n)} \bigg( F^{(2n)}(y_0, \omega_0), G^{(2n)}(y_0, \omega_0) \bigg) \bigg) - \\
&- h \bigg( G^{(2n)} \bigg( F^{(2n)}(y_0, \omega_0), G^{(2n)}(y_0, \omega_0) \bigg) \bigg) \,.
\end{split}\eeq
\end{subequations}
The \eqref{iter_odd_yomega} satisfies the hypothesis of the implicit function theorem up to a set of critical initial conditions $(y_0, \omega_0)$ that do not solve \eqref{iter_odd_yomega} so that an explicit function $\omega_0 = \omega_0(y_0)$ is in general well defined; therefore, the periodic points of odd period must be of the form $p_0 = (x_0(y_0), y_0, z_0(y_0), \omega_0(y_0))$, where $y_0$ is a solution of
\beq\label{Hy0}
y_0 = H^{(2n)}(y_0) \,,
\eeq
with a suitable function $H^{(2n)}$ depending on $F^{(2n)}$ and $G^{(2n)}$ that can be derived from \eqref{iter_odd_yomega}. Given the analytical form of $F^{(2n)}$ and $G^{(2n)}$, we obtain for $y_0$ a functional relationship \eqref{Hy0} analogous to the one that is associated with the fixed points that do not constitute a dense set (see Proposition 1 and Corollary 1 in \cite{Disca}), so we conclude that the map \eqref{map_gh} has not dense sets of periodic points of odd period.
\end{proof}
As an immediate consequence of Proposition \ref{prop_density}, the map \eqref{map} is not chaotic in the sense of Devaney.

It is worth noticing that in the first part of the previous proof, we have not used in any way the analytical form of the functions $g$, $\tilde g$, and $h$, nor their continuity; therefore, we can state the following general result.
\begin{theo}\label{theo1}
Given the functions $g, \tilde g, h \colon \mathbb{R}^4 \to \mathbb{R}^4$, the map
\[
\begin{split}
&f \colon \mathbb{R}^4 \to \mathbb{R}^4 \\
&\overline{x}_{n+1} = f(\overline{x}_n), \quad n \in \mathbb{N}
\end{split}
\]
\[
\begin{split}
&x_{n+1} = y_n \\
&y_{n+1} = g(x_n) + h(z_n) \\
&z_{n+1} = \omega_n \\
&\omega_{n+1} = \tilde g(x_n) - h(z_n)
\end{split}
\]
has not periodic points of even period, i.e., $\forall n \in \mathbb{N}$ $\not\exists p \in \mathbb{R}^4$ such that $f^{(2n)}(p_0) = p_0$.
\end{theo}







\end{document}